\pdfoutput=1
\RequirePackage{ifpdf}
\ifpdf 
\documentclass[pdftex]{sigma}
\else
\documentclass{sigma}
\fi

\usepackage{tikz}
\usepackage{tikzit}

\tikzstyle{nc}=[fill=white, draw=black, shape=circle, minimum size=0.40cm]
\tikzstyle{nr}=[fill=white, draw=black, shape=rectangle, minimum width=0.60cm, minimum height=0.60cm]
\tikzstyle{nrg}=[fill={rgb,255: red,199; green,199; blue,199}, draw=black, shape=rectangle]
\tikzstyle{mr}=[fill=white, draw=black, shape=rectangle, minimum width=1.20cm, minimum height=0.75cm]
\tikzstyle{product}=[fill={rgb,255: red,255; green,5; blue,80}, draw=black, shape=circle, inner sep=0pt, minimum size=0.15cm]
\tikzstyle{unit}=[fill={rgb,255: red,255; green,166; blue,217}, draw=black, shape=circle, inner sep=0pt, minimum size=0.15cm]
\tikzstyle{coproduct}=[fill={rgb,255: red,2; green,145; blue,255}, draw=black, shape=circle, inner sep=0pt, minimum size=0.15cm]
\tikzstyle{counit}=[fill={rgb,255: red,165; green,219; blue,255}, draw=black, shape=circle, inner sep=0pt, minimum size=0.15cm]
\tikzstyle{antipode}=[fill=white, draw=black, shape=circle, inner sep=0pt, minimum size=0.15cm]
\tikzstyle{region}=[fill={rgb,255: red,191; green,191; blue,191}, draw={rgb,255: red,191; green,191; blue,191}, shape=circle, minimum size=0.80cm]
\tikzstyle{boundarydisc}=[fill={rgb,255: red,128; green,128; blue,128}, draw=black, shape=circle, minimum size=0.80cm]
\tikzstyle{identity}=[fill=black, draw=black, shape=circle, inner sep=0pt, minimum size=0.15cm]
\tikzstyle{S}=[fill=white, draw=black, shape=rectangle, minimum width=0.15cm, minimum height=0.15cm]
\tikzstyle{S-}=[fill={rgb,255: red,199; green,199; blue,199}, draw=black, shape=rectangle, minimum width=0.15cm, minimum height=0.15cm]

\tikzstyle{directed}=[draw=black, ->]
\tikzstyle{ds}=[-, dashed]
\tikzstyle{dsd}=[dashed, ->]
\tikzstyle{dc1}=[draw={rgb,255: red,0; green,212; blue,166}, ->]
\tikzstyle{dc2}=[-, draw={rgb,255: red,0; green,212; blue,166}]
\tikzstyle{dd}=[-, dashed, draw={rgb,255: red,0; green,212; blue,166}]
\tikzstyle{p}=[-, draw={rgb,255: red,124; green,95; blue,239}]
\tikzstyle{p2}=[-, dashed, draw={rgb,255: red,124; green,95; blue,239}]
\tikzstyle{rd}=[->, draw={rgb,255: red,255; green,5; blue,80}]
\tikzstyle{rn}=[-, draw={rgb,255: red,255; green,5; blue,80}]
\tikzstyle{bd}=[draw={rgb,255: red,2; green,145; blue,255}, ->]
\tikzstyle{bn}=[-, draw={rgb,255: red,2; green,145; blue,255}]

\def\coloneqq {:=}

\numberwithin{equation}{section}

\newtheorem{Theorem}{Theorem}[section]
\newtheorem{Corollary}[Theorem]{Corollary}
\newtheorem{Lemma}[Theorem]{Lemma}
\newtheorem{Proposition}[Theorem]{Proposition}
 { \theoremstyle{definition}
\newtheorem{Definition}[Theorem]{Definition}
\newtheorem{Remark}[Theorem]{Remark} }

\begin{document}

\newcommand{\arXivNumber}{1911.10147}

\renewcommand{\thefootnote}{}

\renewcommand{\PaperNumber}{040}

\FirstPageHeading

\ShortArticleName{CFT Correlators for Cardy Bulk Fields via String-Net Models}

\ArticleName{CFT Correlators for Cardy Bulk Fields\\ via String-Net Models}

\Author{Christoph SCHWEIGERT and Yang YANG}

\AuthorNameForHeading{C.~Schweigert and Y.~Yang}

\Address{Fachbereich Mathematik, Universit\"at Hamburg, Bereich Algebra und Zahlentheorie,\\
Bundesstra{\ss}e 55, 20146 Hamburg, Germany}
\Email{\href{mailto:christoph.schweigert@uni-hamburg.de}{christoph.schweigert@uni-hamburg.de},
\href{mailto: Yang.Yang@studium.uni-hamburg.de}{Yang.Yang@studium.uni-hamburg.de}}
\URLaddress{\url{https://www.math.uni-hamburg.de/home/schweigert/}}

\ArticleDates{Received November 01, 2020, in final form April 12, 2021; Published online April 21, 2021}

\Abstract{We show that string-net models provide a novel geometric method to construct invariants of mapping class group actions. Concretely, we consider string-net models for a~modular tensor category~${\mathcal C}$. We~show that the datum of a specific commutative symmetric Frobenius algebra in the Drinfeld center~$Z(\mathcal{C})$ gives rise to invariant string-nets. The Frobenius algebra has the interpretation of the algebra of bulk fields of the conformal field theory in the Cardy case.}

\vspace{1mm}


\vspace{1mm}

\Classification{81T40}\vspace{3mm}

\section{Introduction}\vspace{2mm}

Two-dimensional conformal field theories, to which we refer as a CFT
in the following, are~quantum field theories that apart from their
intrinsic physical interest, are amenable to a precise mathematical
study. In this paper, we use string-net models to study consistent
systems of bulk field correlators in a class of such models.

A consistent system of correlators in a CFT is obtained by specifying
elements in spaces of~con\-for\-mal blocks, subject to certain consistency
conditions. For a conformal field theory with the monodromy data given
by a braided monoidal category $\mathcal{D}$, the spaces of conformal
blocks can be constructed as morphism spaces in $\mathcal{D}$. They
are endowed with projective actions of mapping class groups given
in terms of the structures on $\mathcal{D}$. For a rational conformal
field theory, the category $\mathcal{D}$ is a (semisimple) modular
tensor category and the spaces of conformal blocks are provided by
the state spaces of a three-dimensional topological field theory,
namely the Reshetikhin--Turaev TFT based on $\mathcal{D}$. In this
framework, the task of finding a consistent system of correlators
is equivalent to finding for each surface $\Sigma$ a vector in the
space of conformal blocks on the double $\widehat{\Sigma}$. This
element has to be invariant under the action of the mapping class group of~$\Sigma$
and the set of elements has to be consistent under sewing of the surfaces.
This problem has been solved completely in \cite{MR2259258,MR1940282,MR2026879,MR2076134,MR2137114},
using in a non-trivial way the geometry of~certain 3-manifolds. This
is not only technically involved, but also a serious obstacle to extend
the approach to more general classes of CFTs, e.g., those based on
non-semisimple modular tensor categories, since a 3d-TFT of Reshetikhin\textendash Turaev
type with values in vector spaces can only be constructed for semisimple
MTCs.

In this article, we only consider bulk fields on oriented surfaces.
Instead of considering the double $\widehat{\Sigma}$ of the surface
$\Sigma$, which for $\Sigma$ oriented without boundary consists
of two copies of $\Sigma$ with opposite orientation, i.e., $\hat{\Sigma}=\Sigma\sqcup\overline{\Sigma}$
(see, e.g.,~\cite[Section 5.1]{MR1940282}), one uses the following
relation for the state space of the Reshetikhin\textendash Turaev
TFT:
$$
Z_{{\rm RT},\mathcal{C}}(\hat{\Sigma})=Z_{{\rm RT},\mathcal{C}}(\Sigma\sqcup\overline{\Sigma})\cong Z_{{\rm RT},\mathcal{C}}(\Sigma)\otimes Z_{{\rm RT},\mathcal{C}}(\overline{\Sigma})\cong Z_{{\rm RT},\mathcal{C}^{\rm rev}\boxtimes\mathcal{C}}(\Sigma)
$$
 and can take the enveloping category $\mathcal{C}^{\rm rev}\boxtimes\mathcal{C}$
of a modular tensor category as the category $\mathcal{D}$ and stick
with the original surface $\Sigma$. Modularity implies that we have
a braided equivalence: $\mathcal{C}^{\rm rev}\boxtimes\mathcal{C}\simeq Z(\mathcal{C})$,
where $Z(\mathcal{C})$ is the Drinfeld center of $\mathcal{C}$,
see, e.g., \cite{MR3996323} for a statement that includes non-semisimple
categories as well. It~is shown in \cite{balsam2010turaevviro2,balsam2010turaevviro3,KirillovBalsam2010, turaev2010two}
that the Reshetikhin--Turaev construction for $Z(\mathcal{C})$ is
equivalent to the extended Turaev--Viro--Barrett--Westbury state-sum
construction based on $\mathcal{C}$, hence we have
$$
Z_{{\rm RT},\mathcal{C}^{\rm rev}\boxtimes\mathcal{C}}(\Sigma)\cong Z_{{\rm RT},Z(\mathcal{C})}(\Sigma)\cong Z_{{\rm TV},\mathcal{C}}(\Sigma).
$$

The string-net model was first introduced in the study of topological
order in condensed matter physics by Levin and Wen \cite{Levin2005}.
The collection of state spaces associated to surfaces are described
by equivalence classes of string-diagrams on compact oriented surfaces
with boundaries and can be extended to a once-extended TFT which has
recently been shown to be equivalent to the Turaev--Viro--Barrett--Westbury
state-sum construction \cite{Goosen2018Oriented1V,kirillov2011string}.
The string-net model has two advantages that are attractive in our
context: first, a vector in the space of conformal blocks can
be described by a string-net, and second, the action of the mapping
class group, when expressed in terms of such vectors, is completely
geometrical (in fact, the consideration of~mapping class groups actions
on string-nets has already appeared in \cite{koenig2010quantum}).

In this paper, we first define fundamental string-nets on a generating
set of surfaces for every commutative symmetric Frobenius algebra
$F$ in the Drinfeld center $Z(\mathcal{C})$, using the structure
morphisms of $F$. We~show in Lemma~\ref{lem:F} that those string-nets
are invariant under the mapping class group action. Moreover, the
prescription extends to a consistent system of correlators in the
sense of \cite{MR3590526} by sewing, where the Frobenius algebra
$F$ befits the algebra of bulk fields, provided that the string-net
on the torus with one boundary circle is invariant under the mapping
class group action. It~can be inferred from the known result~\cite[Theorem 3.4]{MR2551797}
that a~haploid commutative symmetric Frobenius algebra $F\in Z(\mathcal{C})$
satisfies this condition if and only if~$\dim (F)$ equals the
global dimension of the category~$\mathcal{C}$. Then to each surface~$\Sigma$, possibly with non-empty boundary, the assigned correlator
can be obtained as a string-net by decomposing the surface into pairs
of pants and placing the appropriate fundamental string-nets on each
component. For instance, for a surface of genus one with one ingoing
and two outgoing boundary components, we have the following string-net
\begin{figure}[h!]
\centering
$\input{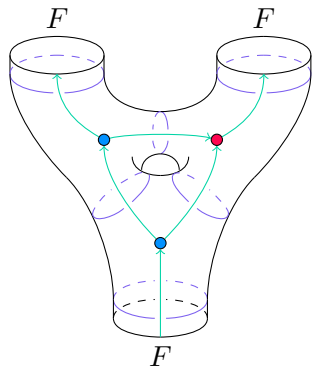}$
\caption{\label{fig:example}The string-net assigned to the extended surface
of genus one with one ingoing and two out\-going boundary circles according
to a certain pairs of pants decomposition.}
\end{figure}

{\samepage
Here the green lines are labeled by the Frobenius algebra
$F$, the red and blue circular coupons stand for the multiplication
and co-multiplication of the Frobenius algebra $F$ respectively,
while the purple circles stand for the boundary projectors (introduced
in Remark~\ref{rem:proj}) that account for the half-braiding of $F$.

}

We then restrict to a specific algebra $F_{1}$ satisfying the condition
of modular invariance: the bulk algebra for the Cardy case in which
the modular invariant on a torus is given by the charge conjugation
matrix. The underlying object of the algebra $F_{1}$ is
$$
L=\bigoplus_{i\in\mathcal{I}(\mathcal{C})}X_{i}^{\vee}\otimes X_{i}\in\mathcal{C}
$$
 along with a certain half-braiding (see Section~\ref{subsec:Cardy}).
Here $\mathcal{I}(\mathcal{C})$ stands for the set of isomorphism
classes of simple objects in $\mathcal{C}$ and $X_{i}$ is a fixed
representative for each $i\in\mathcal{I}(\mathcal{C})$. We~show that
for the algebra $F_{1}$, the string-nets describing the correlators
are almost empty (Theo\-rem~\ref{thm:main}). For instance, the string-net
shown in Figure~\ref{fig:example}, after substituting the algebra
with~$F_{1}$, will be shown to be the following string-net
\begin{figure}[h!]
\centering{}$\displaystyle \sum_{i,j,k,l,m,n\in\mathcal{I}(\mathcal{C})}\frac{d_{j}d_{k}d_{l}d_{m}d_{n}}{D^{6}}\;\input{IN2.tikz}$\caption{The simplified form of the string-net assigned to the extended surface
$\Sigma_{1\mid2}^{1}$. }
\end{figure}

These correlators have been constructed in terms of the evaluation
of a 3d-TFT on certain ribbon graphs in 3-manifolds in~\cite{Felder:1999mq}.
The geometry and the ribbon graphs are quite involved. The fact that
we can describe them by almost empty string-nets demonstrates the
advantage of~the string-net construction.

This paper is organized as follows: in Section~\ref{sec:String-net-model},
we briefly review string-net models, follo\-wing~\cite{kirillov2011string}. We~next recall some facts about modular tensor categories in Section~\ref{subsec:Modular-tensor-categories}, review the notion of a consistent
system of bulk field correlators in Section~\ref{subsec:CSOC}. We~define the fundamental string-nets in Section~\ref{subsec:Fund}
and show that they give rise to a consistent system of correlators
when the used Frobenius algebra is modular. Section~\ref{sec:Cardy-case}
is devoted to the Cardy case.

We expect that our results can be generalized in several directions:
beyond the Cardy case and to correlators including also boundary and
defect fields. A generalization of the string net construction to
non-semisimple finite tensor categories remains, at the moment, a
challenge. It~would allow us to address correlators of logarithmic
conformal field theories as well in a~two-dimensional setting.

\section{String-net models \label{sec:String-net-model}}

\subsection{Spherical fusion categories}

String-net models are defined for spherical fusion categories. In
this section, we review some basic facts of spherical fusion categories
and fix our notations. We~denote by~$\mathbb{K}$ an algebraically
closed field of characteristic~0.

Recall that a \textit{right dual} of an object $V$ in a strict monoidal
category $\mathcal{C}$ is an object $V^{\vee}$ together with morphisms
$\mathrm{coev}_{V}\in\mathrm{Hom}_{\mathcal{C}}\big(\mathbb{I},V\otimes V^{\vee}\big)$
and $\mathrm{ev}_{V}\in\mathrm{Hom}_{\mathcal{C}}\big(V^{\vee}\otimes V,\mathbb{I}\big)$
satisfying
$$
(\mathrm{id}_{V}\otimes\mathrm{ev}_{V})\circ(\mathrm{coev}_{V}\otimes\mathrm{id}_{V})=\mathrm{id}_{V}
$$
and
$$
(\mathrm{ev}_{V}\otimes\mathrm{id}_{V^{\vee}})\circ(\mathrm{id}_{V^{\vee}}
\otimes\mathrm{coev}_{V})=\mathrm{id}_{V^{\vee}}.
$$
We depict the right duality maps graphically as
\begin{center}
$\input{GC7.tikz}=\input{GC9.tikz}\,, \qquad\qquad \input{GC8.tikz}=\input{GC10.tikz}\!.$
\end{center}

\noindent Here we replaced $V^{\vee}$ by $V$ upon reversing the
direction of the arrow. Left duality is defined similarly by reversing
the arrows in the graphical notation. A monoidal category in which
every object has both left and right duals is called a \textit{rigid
monoidal category}.

{\sloppy
A \textit{pivotal structure} on a rigid monoidal category is a monoidal
natural isomorphism \mbox{$\omega\colon\mathrm{id}_{\mathcal{C}}\Rightarrow(-)^{\vee\vee}$}.
A pivotal structure is called strict if $\mathrm{id}_{\mathcal{C}}=(-)^{\vee\vee}$
and $\omega=\mathrm{id}_{\mathcal{\mathrm{id}_{\mathcal{C}}}}$. It~is known that every pivotal category is pivotally equivalent to a
pivotal category with strict pivotal structure~\cite[Theorem 2.2]{MR2381536},
hence we will assume the pivotal structure to be strict in the following
without loss of generality. For a strict pivotal category, the left
and right duality strictly coincide as functors.

}

In a pivotal category we have the notions of \textit{right and left
traces} for any $f\in\mathrm{End}_{\mathcal{C}}(V)$. Graphically
\begin{center}
$\mathrm{tr}_{\mathrm{r}}(f)\coloneqq\input{GC17.tikz}\in\mathrm{End}_{\mathcal{C}}(\mathbb{I}), \qquad\qquad \mathrm{tr}_{\mathrm{l}}(f)\coloneqq\input{GC18.tikz}\in\mathrm{End}_{\mathcal{C}}(\mathbb{I})$.
\end{center}

\noindent When applied to $\mathrm{id}_{V}\in\mathrm{End}_{\mathcal{C}}(V)$,
we get the definitions of the \textit{left and right categorical dimension}
of the object $V\in\mathcal{C}$. A pivotal category is called \textit{spherical}
if the left and right traces coincide, i.e., $\mathrm{tr}(f)\coloneqq\mathrm{tr}_{\mathrm{r}}(f)=\mathrm{tr}_{\mathrm{l}}(f)$
and $\dim (V)\coloneqq\dim _{\mathrm{r}}(V)=\dim _{\mathrm{l}}(V)$.
\begin{Definition}
A \textit{fusion category} over $\mathbb{K}$ is a rigid $\mathbb{K}$-linear
monoidal category $\mathcal{C}$ that is finitely semisimple, with
the monoidal unit $\mathbb{I}$ being simple. A \textit{spherical
fusion category} over $\mathbb{K}$ is a~sphe\-rical category $\mathcal{C}$
that is also a fusion category over $\mathbb{K}$.
\end{Definition}

\looseness=1
Here being $\mathbb{K}$-linear means that the sets of morphisms are
$\mathbb{K}$-vector spaces and the composition as well as the monoidal
product are bilinear. Being finitely semisimple means that there are
finitely many isomorphism classes of simple objects (objects with
no non-trivial subobject) and every object is a direct sum of finitely
many simple objects. Note that $\mathbb{K}$-linearity and finite-semisimplicity
together imply that the morphism spaces are finite dimensional.

Let us denote the set of isomorphism classes of simple objects by
$\mathcal{I}(\mathcal{C})$, and fix a repre\-sen\-ta\-tive~$X_{i}$ for
each $i\in\mathcal{I}(\mathcal{C})$. In addition, we require $0\in\mathcal{I}(\mathcal{C})$
and $X_{0}=\mathbb{I}$. Duality furnishes a~invo\-lu\-tion on $\mathcal{I}(\mathcal{C})$,
i.e., $i\mapsto\bar{i}\coloneqq\big[X_{i}^{\vee}\big]$. We~require that $X_{\bar{i}}=X_{i}^{\vee}$
whenever $i\neq\bar{i}$. Since $\mathbb{K}$ is assumed to be algebraically
closed, the only finite dimensional division algebra over $\mathbb{K}$
is $\mathbb{K}$ itself. Thus we have Schur's lemma: $\mathrm{Hom}_{\mathcal{C}}(X_{i},X_{j})\cong\delta_{i,j}\mathbb{K}.$
In particular, $d_{X}\coloneqq\dim (X)\in\mathrm{End}_{\mathcal{C}}(\mathbb{I})\cong\mathbb{K}$.
Define the \textit{global dimension} of the spherical fusion category
$\mathcal{C}$ to be
$$
D^{2}\coloneqq\sum_{i\in\mathcal{I}(\mathcal{C})}d_{i}^{2}.
$$
Despite of the notation, we do not choose a square root of the global
dimension. By~\cite[Theorem~2.3]{MR2183279}, $D^{2}\ne0$.

We define the functor $\underbrace{\mathcal{C}\boxtimes\cdots\boxtimes\mathcal{C}}_{n}\to\mathcal{V}\text{ect}_{\mathbb{K}}$
by
$$
V_{1}\boxtimes\cdots\boxtimes V_{n}\mapsto\mathrm{Hom}_{\mathcal{C}}(\mathbb{I},V_{1}\otimes\cdots\otimes V_{n}).
$$
The pivotal structure furnishes a natural isomorphism by
$$
\begin{array}{rcl}
z_{V_{1}\boxtimes\cdots\boxtimes V_{n}}\colon\ \mathrm{Hom}_{\mathcal{C}}(\mathbb{I},V_{1},\dots,V_{n})&\xrightarrow{\cong}&
\mathrm{Hom}_{\mathcal{C}}(\mathbb{I},V_{n},V_{1},\dots,V_{n-1}),
\\[2ex]
\input{SF1.tikz}\phantom{A}&\mapsto&\phantom{A}\input{SF2.tikz}
\end{array}
$$
It can be seen that $z^{n}=\mathrm{id}$. Thus, up to a
natural isomorphism, $\mathrm{Hom}_{\mathcal{C}}(\mathbb{I},V_{1},\dots,V_{n})$
depends~only on the cyclic order of $V_{1},\dots,V_{n}$. This allows
us to represent an element $\varphi\in$ $\mathrm{Hom}_{\mathcal{C}}(\mathbb{I},V_{1},\dots,V_{n})$
by a round coupon with $n$ outgoing legs colored by $V_{1},\dots,V_{n}$
in clockwise order
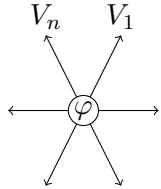
\begin{center}
$\input{SF3.tikz}$
\end{center}

\noindent
We are able to connect legs with dual labels: define the composition
map
$$
\begin{array}{rcl}
\mathrm{Hom}_{\mathcal{C}}\big(\mathbb{I},V_{1},\dots,V_{n},X^{\vee}\big)
\!\otimes_{\mathbb{K}}\!\mathrm{Hom}_{\mathcal{C}}(\mathbb{I},X,W_{1},\dots,W_{m})
\!&\to&\mathrm{Hom}_{\mathcal{C}}(\mathbb{I},V_{1},\dots,V_{n},W_{1},\dots,W_{m}),
\\[1ex]
\varphi\otimes_{\mathbb{K}}\psi&\mapsto&\varphi\circ_{X}\psi\coloneqq\mathrm{ev}_{X}
\circ(\varphi\otimes_{\mathbb{K}}\psi).
\end{array}
$$
This gives rise to a pairing: $\mathrm{Hom}_{\mathcal{C}}(\mathbb{I},V_{1},\dots,V_{n}) \otimes_{\mathbb{K}}\mathrm{Hom}_{\mathcal{C}}\big(\mathbb{I},V_{n}^{\vee},\dots,V_{1}^{\vee}\big)\to\mathbb{K}$. It~is nondegenerate due to the nondegeneracy of the evaluation maps.
Hence for any choice of bases $\{\varphi_{\alpha}\}_{\alpha\in A}$
of $\mathrm{Hom}_{\mathcal{C}}(\mathbb{I},V_{1},\dots,V_{n})$, we
define the dual bases $\{\varphi^{\alpha}\}_{\alpha\in A}$ with respect
to this nondegenerate pairing. In the following we will use the following
summation convention\vspace{2ex}
\begin{center}
$\input{SF4.tikz}\displaystyle\coloneqq\sum_{\alpha\in A}\varphi^{\alpha}\otimes_{\mathbb{K}}\varphi_{\alpha}$.
\end{center}

\noindent Such expressions are independent of the choice of bases.

We now introduce the following useful \textit{completeness relation}:
\begin{Proposition}\label{prop:useful}
For any $V_{1},\dots,V_{n}\in\mathcal{C}$, we have
\begin{center}
$\displaystyle\sum_{i\in\mathcal{I}(\mathcal{C})}d_{i}\input{SF5.tikz}=\input{SF6.tikz}$
\end{center}
\end{Proposition}

\subsection{The string-net construction}

We now give a brief introduction to the string-net construction. We~refer to \cite{kirillov2011string} for more details, and to \cite{Levin2005}
for motivations from physics.

Let's consider \textit{finite graphs} (i.e., the sets of the vertices
and the edges are both finite) embedded in an oriented surface $\Sigma$, which is not required to be compact and may have non-empty boundary.
For such a graph $\Gamma$, denote by $\mathcal{E}^{\rm or}(\Gamma)$
the set of its \textit{oriented edges} and $\mathcal{V}(\Gamma)$
the set of~its \textit{vertices}. One-valent vertices are called \textit{endings}. We~denote the set of endings of~$\Gamma$ by~$\mathcal{V}^{\rm en}(\Gamma)$,
and define $\mathcal{V}^{\rm in}(\Gamma)\coloneqq\mathcal{V}(\Gamma)\setminus\mathcal{V}^{\rm en}(\Gamma)$. We~require $\Gamma\cap\partial\Sigma=\mathcal{V}^{\rm en}(\Gamma)$. We~will call the edges terminating at endings \textit{legs}. Note that
we don't make a choice of orientations for the edges of~the finite
graphs.
\begin{Definition}
Let $\mathcal{C}$ be a spherical fusion category, $\Sigma$ and $\Gamma$
be as defined above. A $\mathcal{C}$-\textit{coloring} (or simply
\textit{coloring} when there is no ambiguity) of $\Gamma$ is given
by the following data:
\begin{itemize}\itemsep=0pt
\item A map $V\colon\mathcal{E}^{\rm or}(\Gamma)\to\mathrm{Obj}(\mathcal{C})$
such that for every $e\in\mathcal{E}^{\rm or}(\Gamma)$, we have $V(e)=V(\overline{e})^{*}$,
where~$\overline{e}$ is the edge with opposite orientation of $e$.
\item A choice of a vector $\varphi(v)\in\mathrm{Hom}_{\mathcal{C}}(\mathbb{I},V(e_{1}),\dots,V(e_{n}))$
for every $v\in\mathcal{V}^{\rm in}(\Gamma)$, where $e_{1},\dots,e_{n}$
are incident to $v$, taken in clockwise order (when the orientation
of the surface is considered conterclockwise) and with outward orientation.
\end{itemize}
An \textit{isomorphism} $f$ of two colorings $(V,\varphi)$ and $(V',\varphi')$
is a collection of isomorphisms $f_{e}\colon V(e)\xrightarrow{\cong}V'(e)$
that is compatible with $V(e)=V(\overline{e})^{*}$ and such that
$\varphi'(v)=\big({\bigotimes}_{e\in\mathcal{E}^{\rm or}(v)}f_{e}\big)\circ\varphi(v)$,
where~$\mathcal{E}^{\rm or}(v)$ is the set of edges that are incident to the vertex~$v$.
\end{Definition}

Let $B\subset\partial\Sigma$ be a finite collection of points on
$\partial\Sigma$ and $\mathbf{V}\colon B\to\mathrm{Obj}(\mathcal{C})$
a map. A $\mathcal{C}$-colored graph $\Gamma$ with \textit{boundary
value} $\mathbf{V}$ is a colored graph such that $\mathcal{V}^{\rm en}(\Gamma)=B$
and $V(e_{b})=\mathbf{V}(b)$, where $b\in B$ and $e_{b}$ is the
edge incident to $b$ with outgoing orientation. We~define $\mathrm{Graph}(\Sigma,\mathbf{V})$
to be the set of $\mathcal{C}$-colored graphs in $\Sigma$ with boundary
value $\mathbf{V}$, and $\mathrm{VGraph}(\Sigma,\mathbf{V})$ to
be the $\mathbb{K}$-vector space freely generated by $\mathrm{Graph}(\Sigma,\mathbf{V})$.

When $\Sigma$ happens to be a disc $D\subset\mathbb{R}^{2}$, a colored
graph $\Gamma\in$$\mathrm{Graph}(D,\mathbf{V})$ can be naturally
viewed as the graphical representation of some morphism in $\mathcal{C}$.
Indeed, graphical calculus for spherical fusion categories provides
a canonical linear surjection~\cite[Theorem 2.3]{kirillov2011string}
$$
\left\langle -\right\rangle _{D}\colon\ \mathrm{VGraph}(\Sigma,\mathbf{V})\to\mathrm{Hom}_{\mathcal{C}}(\mathbb{I},V(e_{1}),\dots,V(e_{n})),
$$
 where $B=\{b_{1},\dots,b_{n}\}$ and $e_{1},\dots,e_{n}$ are the
corresponding outgoing legs, taken in the clockwise order.

The finite dimensional vector space $\mathrm{Hom}_{\mathcal{C}}(\mathbb{I},V(e_{1}),\dots,V(e_{n}))\cong\mathrm{VGraph}(D,\mathbf{V})/\mathrm{ker}\left\langle -\right\rangle _{D}$
can be viewed as the space of linear combinations of $\mathcal{C}$-colored
graphs with a fixed boundary value, where two combinations are identified
if they represent the same morphism in $\mathcal{C}$ according to
the graphical calculus. The identification in turn allows us to perform
graphical calculus in~this space. This inspires us to use $\mathrm{VGraph}(D,\mathbf{V})/\mathrm{ker}\left\langle -\right\rangle _{D}$
as a local model to define a vector space for an arbitrary oriented
surface $\Sigma$ with a prescribed boundary value $\mathbf{V}$,
so that we can perform graphical calculus locally.
\begin{Definition}
Let $D\subset\Sigma$ be an embedded disc. A \textit{null graph with
respect to $D$} is a linear combination of colored graphs $\boldsymbol{\Gamma}=c_{1}\Gamma_{1}+\cdots+c_{n}\Gamma_{n}\in\mathrm{VGraph}(\Sigma,\mathbf{V})$
such that
\begin{itemize}\itemsep=0pt
\item $\boldsymbol{\mathbf{\Gamma}}$ is transversal to $\partial D$ (i.e.,
no vertex of $\Gamma_{i}$ is on $\partial D$ and the edges of each
$\Gamma_{i}$ intersect with $\partial D$ transversally).
\item All $\Gamma_{i}$ coincide outside of $D$.
\item $\left\langle \boldsymbol{\Gamma}\right\rangle _{D}={\displaystyle \sum_{i}c_{i}\left\langle \Gamma_{i}\cap D\right\rangle _{D}}=0$.
\end{itemize}
Denote by $\mathrm{N}(\Sigma,\mathbf{V})\subset\mathrm{VGraph}(\Sigma,\mathbf{V})$
the subspace spanned by null graphs for all possible embedded disks
$D\subset\Sigma$.
\end{Definition}

\begin{Definition}
Let $\Sigma$ be an oriented surface and
let $\mathbf{V}\colon B\to\mathrm{Obj}(\mathcal{C})$ be a boundary
value. Define the \textit{string-net space} for $(\Sigma,\mathbf{V})$
to be the quotient space
$$
Z_{{\rm SN},\mathcal{C}}(\Sigma,\mathbf{V})\coloneqq\mathrm{VGraph}(\Sigma,\mathbf{V})/\mathrm{N}(\Sigma,\mathbf{V}).
$$
\end{Definition}

As before, we have a linear surjection
$$
\left\langle -\right\rangle _{\Sigma}\colon\ \mathrm{VGraph}(\Sigma,\mathbf{V})\to Z_{{\rm SN},\mathcal{C}}(\Sigma,\mathbf{V}).
$$
 The map has several nice properties. For instance, it is linear in
the colors of vertices and additive with respect to direct sums, isotopic
graphs and graphs with isomorphic colorings have the same image. But
most importantly, it allows us to identify graphs that only differ
by local relations that are encoded by the definition of null graphs.
That is to say, all equations from the graphical calculus for the
spherical fusion category $\mathcal{C}$, e.g., the one from Proposition~\ref{prop:useful}, holds true inside any embedded disc on the surface.

\subsection{Drinfeld center and the extended string-net spaces}

One can associate to any monoidal category $\mathcal{C}$ a braided
monoidal category $Z(\mathcal{C})$, called the Drinfeld center of
$\mathcal{C}$. Recall that the objects of the Drinfeld center $Z(\mathcal{C})$
are given by the pairs $Y=(\dot{Y},\gamma_{Y})$, where $\dot{Y}\in\mathcal{C}$
and $\gamma_{Y}\colon\dot{Y}\otimes-\Rightarrow-\otimes\dot{Y}$ is
a natural isomorphism called the \textit{half-braiding} subjected
to certain conditions. We~use the over-crossing of a green line labeled
by an~object $Y\in Z(\mathcal{C})$ to denote its half-braiding
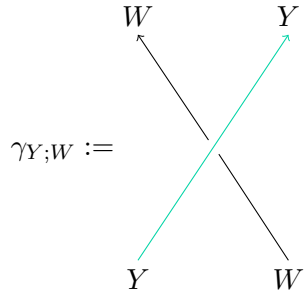
\begin{center}
$\gamma_{Y;W}:={\input{E1.tikz}}$
\end{center}

The definition of string-net spaces can be modified so that one assign
to each boundary circle an object in the Drinfeld center $Z(\mathcal{C})$. We~now give a working description of the extended string-net spaces
that are relevant to our construction of CFT correlators and refer
to~\cite[Sections~6 and~7]{kirillov2011string} for details.
\begin{Remark}\label{rem:proj}
Let $S^{1}$ be an oriented circle. One can define
a $\mathbb{K}$-linear category $\hat{\mathcal{C}}\big(S^{1}\big)$ whose
objects are oriented circles with finite numbers of points labeled
by objects of $\mathcal{C}$ and whose morphism space between two
such circles are the string-net space on a cylinder with boundary
value induced by the inclusion of the two circles as the boundary
of the cylinder. The composition of morphisms are induced by the stacking
of cylinders and the concatenation of string-nets, see~\cite[Definition~6.1]{kirillov2011string}
for details. For all $Y\in Z(\mathcal{C})$, the following string-net
on a cylinder, considered as a morphism in $\hat{\mathcal{C}}\big(S^{1}\big)$,
is a projector
\begin{center}
$P_Y\coloneqq{\displaystyle \sum_{i\in\mathcal{I}(\mathcal{C})}\frac{d_{i}}{D^{2}}}\input{DB7.tikz}$
\end{center}

\noindent This can be seen by the following calculation in the string-net
space, using the completeness relation in Proposition~\ref{prop:useful}
and the naturality of the half-braiding
$$
P_Y^{2}=\sum_{i,j\in\mathcal{I}(\mathcal{C})}\frac{d_{i}d_{j}}{D^{4}}\input{DB8.tikz}
=\sum_{i,j,k\in\mathcal{I}(\mathcal{C})}\frac{d_{i}d_{j}d_{k}}{D^{4}}\input{DB9.tikz}
=\sum_{j,k\in\mathcal{I}(\mathcal{C})}\frac{d_{j}d_{k}}{D^{4}}\input{DB10.tikz}$$
$$=\sum_{k\in\mathcal{I}(\mathcal{C})}\frac{d_{k}}{D^{2}}\phantom{P}\input{DB11.tikz}=P_Y.
$$

We denote by an unoriented, unlabeled purple line the following morphism
that is sometimes called the \textit{canonical color}, the \textit{Kirby
color}, or the \textit{surgery color}
\begin{center}
$\begin{tikzpicture}
	\begin{pgfonlayer}{nodelayer}
		\node [style=none] (0) at (0, 2) {};
		\node [style=none] (1) at (0, -2) {};
	\end{pgfonlayer}
	\begin{pgfonlayer}{edgelayer}
		\draw [style=p] (0.center) to (1.center);
	\end{pgfonlayer}
\end{tikzpicture}
\coloneqq{\displaystyle \sum_{i\in\mathcal{I}(\mathcal{C})}\frac{d_{i}}{D^{2}}}\begin{tikzpicture}
	\begin{pgfonlayer}{nodelayer}
		\node [style=none] (0) at (0, 2) {};
		\node [style=none] (1) at (0, -2) {};
		\node [style=none] (2) at (0, -2.5) {$i$};
	\end{pgfonlayer}
	\begin{pgfonlayer}{edgelayer}
		\draw [style=directed] (1.center) to (0.center);
	\end{pgfonlayer}
\end{tikzpicture}

\in\mathrm{End}_{\mathcal{C}}\big({\bigoplus}_{i\in\mathcal{I}(\mathcal{C})}X_{i}\big).$
\end{center}
\end{Remark}

Therefore, the projector $P_{Y}$ can be also expressed as
\begin{center}
$P_Y= \input{E4.tikz}$
\end{center}

We are interested in the case where $\Sigma\cong\Sigma_{n}^{g}$,
here $\Sigma_{n}^{g}$ means a compact oriented surface of~genus~$g$
with $n$ boundary components. Denote by $\big(\Sigma_{n}^{g},Y_{1},\dots,Y_{n}\big)$
a \textit{$Z(\mathcal{C})$-marked surface}, i.e., $\Sigma_{n}^{g}$
together with
\begin{itemize}\itemsep=0pt
\item a numbering of $\pi_{0}(\partial\Sigma)$ with $1,\dots,n$,
\item a choice of a point in each connected component of $\partial\Sigma$,
\item a choice of $n$ objects $Y_{1},\dots,Y_{n}\in Z(\mathcal{C})$.
\end{itemize}
We denote the extended string-net space for the $Z(\mathcal{C})$\textit{-}marked
surface $\Sigma_{n}^{g}$ with this boundary value by $Z_{{\rm SN},\mathcal{C}}\big(\Sigma_{n}^{g},Y_{1},\dots,Y_{n}\big)$.
This is defined to be a subspace of the (unextended) string-net spaces
of $\Sigma_{n}^{g}$ with boundary value given by the underlying objects
of $Y_{1},\dots,Y_{n}$ in $\mathcal{C}$, which is spanned by all
the string-nets with the additional projectors introduced in Remark~\ref{rem:proj} placed near the corresponding boundary circles. For
instance, a generic vector in $Z_{{\rm SN},\mathcal{C}}\big(\Sigma_{3}^{1},Y_{1},Y_{2},Y_{3}\big)$
can be defined by a linear combination of equivalence classes of colored
graphs on $\Sigma_{3}^{1}$ such as
\begin{figure}[h!]
\centering
$\input{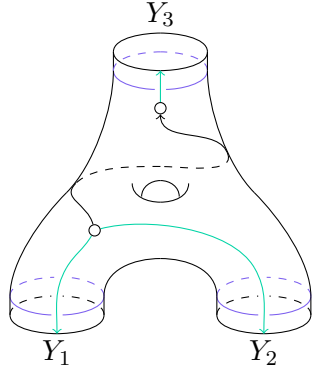}$
\caption{A generic string-net in $Z_{{\rm SN},\mathcal{C}}(\Sigma_{3}^{1},Y_{1},Y_{2},Y_{3}$).}
\end{figure}

There is a canonical isomorphism
$$
Z_{{\rm SN},\mathcal{C}}\big(\Sigma_{n}^{g},Y_{1},\dots,Y_{n}\big)\cong Z_{{\rm TV},\mathcal{C}}\big(\Sigma_{n}^{g},Y_{1},\dots,Y_{n}\big),
$$
where
$Z_{{\rm TV},\mathcal{C}}\big(\Sigma_{n}^{g},Y_{1},\dots,Y_{n}\big)$ is
the state space for $\big(\Sigma_{n}^{g},Y_{1},\dots,Y_{n}\big)$ in the
extended Turaev--Viro--Bar\-rett--Westbury topological field theory~\cite{kirillov2011string}.
Hence:
\begin{Proposition}
There are isomorphisms
$$
Z_{{\rm SN},\mathcal{C}}\big(\Sigma_{n}^{g},Y_{1},\dots,Y_{n}\big)\cong Z_{{\rm TV},\mathcal{C}}\big(\Sigma_{n}^{g},Y_{1},\dots,Y_{n}\big)\cong\mathrm{Hom}_{Z(\mathcal{C})} \big(\mathbb{I}_{Z(\mathcal{C})},Y_{1}\otimes\cdots\otimes Y_{n}\otimes L_{Z(\mathcal{C})}^{\otimes g}\big)
$$
 that are functorial with respect to the morphisms in $Z(\mathcal{C})$,
where $L_{Z(\mathcal{C})}\coloneqq\bigoplus_{i\in\mathcal{I}(Z(\mathcal{C}))}Z_{i}^{\vee}\otimes Z_{i}$.
\end{Proposition}

\section{Bulk field correlators}

\subsection{Modular tensor categories\label{subsec:Modular-tensor-categories}}

The categorical ingredient of the string-net construction is a spherical
fusion category $\mathcal{C}$, which is not necessarily braided.
However, for the application to conformal field theories, we need
a~category with the structure of a ribbon fusion category over $\mathbb{C}$
with an additional nondegeneracy property:
\begin{Definition}
A \textit{modular tensor category} $\mathcal{C}$ is a ribbon fusion
category over $\mathbb{C}$ with the braiding being nondegenerate
in the sense that the matrix $(s_{i,j})_{i,j\in\mathcal{I}(\mathcal{C})}$
is invertible, where $s_{i,j}\coloneqq\mathrm{tr}(\beta_{j,i}\circ\beta_{i,j})$.
\end{Definition}

It can be seen from the cyclic symmetry of the categorical trace that
$s_{i,j}=s_{j,i}$. Moreover, one can show that (see, e.g., \cite[Theorem~3.7.1]{MR1797619})
$$
\sum_{k\in\mathcal{I}(\mathcal{C})}s_{i,k}s_{k,j}=D^{2}\delta_{i,\bar{j}}.
$$

For a spherical fusion category $\mathcal{C}$, the Drinfeld center
$Z(\mathcal{C})$ is a modular tensor category.

We are going to use the following lemmas:

\begin{Lemma}
For every $i\in\mathcal{I}(\mathcal{C})$, we have
\begin{center}
$\input{MTC2.tikz}=\dfrac{s_{j,i}}{d_{i}}\;\begin{tikzpicture}
	\begin{pgfonlayer}{nodelayer}
		\node [style=none] (10) at (0, 3) {};
		\node [style=none] (11) at (0, -3) {};
		\node [style=none] (14) at (0.25, 0) {$i$};
	\end{pgfonlayer}
	\begin{pgfonlayer}{edgelayer}
		\draw [style=directed] (11.center) to (10.center);
	\end{pgfonlayer}
\end{tikzpicture}
 \qquad and \qquad \displaystyle\sum_{j\in\mathcal{I}(\mathcal{C})}d_{j}\;\input{MTC2.tikz}=D^{2}\delta_{0,i}\;\begin{tikzpicture}
	\begin{pgfonlayer}{nodelayer}
		\node [style=none] (10) at (0, 3) {};
		\node [style=none] (11) at (0, -3) {};
		\node [style=none] (14) at (0.25, 0) {$i$};
	\end{pgfonlayer}
	\begin{pgfonlayer}{edgelayer}
		\draw [style=directed] (11.center) to (10.center);
	\end{pgfonlayer}
\end{tikzpicture}
$\,.
\end{center}
\end{Lemma}

{\samepage\begin{Lemma}
\label{lem:mtc3}For every $X\in\mathcal{C}$, we have
\begin{center}
$\displaystyle\sum_{i\in\mathcal{I}(\mathcal{C})}d_{i}\;\input{MTC5.tikz}=D^{2}\;\input{MTC6.tikz}$
\end{center}

\noindent In particular, for every $i,j\in\mathcal{I}(\mathcal{C})$,
we have
\begin{center}
$\displaystyle\sum_{k\in\mathcal{I}(\mathcal{C})}d_{k}\;\input{MTC7.tikz}
=\dfrac{D^{2}}{d_{i}}\delta_{i,j}\;\input{MTC8.tikz}$
\end{center}

\end{Lemma}

}

For a ribbon category $\mathcal{C}$, we denote by $\mathcal{C}^{\rm rev}$
its reverse category, i.e., the same monoidal category with inverse
braiding and twist. There is a canonical braided functor
$$\begin{array}{rcl}
\Xi\colon\ \mathcal{C}^{\rm rev}\boxtimes\mathcal{C}&\to& Z(\mathcal{C}),
\\[1ex]
U\boxtimes V&\mapsto&\big(U\otimes V,\gamma_{U\otimes V}^{\mathrm{uo}}\big),
\end{array}$$
 where
\begin{center}
$\gamma_{U\otimes V;W}^{\mathrm{uo}}\coloneqq\;\scalebox{0.8}{\input{MTC4.tikz}}$
\end{center}

\noindent is the \textit{under-over half-braiding}.

In fact, modularity can be formulated as follows, see, e.g.,~\cite{MR3996323}:
\begin{Proposition}\label{prp:mtc}A ribbon fusion category $\mathcal{C}$ is a modular
tensor category if and only if the canonical functor $\Xi$ is a braided
equivalence.
\end{Proposition}

\subsection{Consistent systems of correlators} \label{subsec:CSOC}

We now give a summary of the concept of consistent systems of CFT
bulk field correlators in~the form used in \cite{MR3590526}.

An \textit{extended surface} $\Sigma$ is an oriented surface with
a partition of the boundary components into ingoing and outgoing parts,
i.e., $\partial\Sigma=\partial_{\rm in}\Sigma\sqcup\partial_{\rm out}\Sigma$
and a marked point for each boundary component. We~denote by $\Sigma_{p\mid q}^{g}$
an extended surface of genus $g$ with $p$ ingoing boun\-dary components
and $q$ outgoing boundary components.
\begin{Definition}
The \textit{mapping class group} $\mathrm{Map}(\Sigma)$ of an extended
surface $\Sigma$ is the group of~homo\-topy classes of orientation
preserving homeomorphisms $\Sigma\to\Sigma$ that map $\partial_{\rm in}\Sigma$
to itself (hence also $\partial_{\rm out}\Sigma$ to itself) and map marked
points to marked points.
\end{Definition}

Along with the action of the mapping class group on an extended surface,
we also consider the \textit{sewing} of the surface: A sewing $s_{\alpha,\beta}$
along $(\alpha,\beta)\in\pi_{0}(\partial_{\rm in}\Sigma)\times\pi_{0}(\partial_{\rm out}\Sigma)$
gives us a new extended surface $s_{\alpha,\beta}(\Sigma)\coloneqq\cup_{\alpha,\beta}\Sigma$
by identifying the boundary component $\partial_{\alpha}\Sigma$ with
$\partial_{\beta}\Sigma$ via an orientation preserving homeomorphism
$f\colon\partial_{\alpha}\Sigma\to\partial_{\beta}\Sigma$ that maps
the marked point on $\partial_{\alpha}\Sigma$ to the marked point
on $\partial_{\beta}\Sigma$. The resulting surface is independent
of $f$ up to homeomorphisms.
\begin{Definition}
The category $\mathcal{S}\text{urf}$ is the symmetric monoidal category
having extended sur\-faces~$\Sigma$ as objects and the pairs $(\varphi,s_{\alpha,\beta})$
as morphisms $\Sigma\to\cup_{\alpha,\beta}\Sigma$, where $\varphi\in\mathrm{Map}(\Sigma)$
is a~mapping class and $s_{\alpha,\beta}$ a sewing. The monoidal
product is given by the disjoint union.
\end{Definition}

In order to describe the composition of the morphisms in the category
$\mathcal{S}\text{urf}$, we need the relations among the pairs of
mapping classes and sewing. Such relations are discussed in detail
in \cite{MR1752293}.

In a local two-dimensional conformal field theory, specific spaces
of conformal blocks for bulk fields can be constructed as the morphism
spaces in a braided monoidal category $\mathcal{D}$ involving a fixed
object $F\in\mathcal{D}$. The object $F\in\mathcal{D}$ should be
considered as the space of bulk fields. We~say that the CFT has the
\textit{monodromy data} based on $\mathcal{D}$ and the \textit{bulk
object} $F$. The reader is invited to think of $\mathcal{D}$ as
the representation category of both left moving and right moving chiral
symmetries. The collection of all bulk fields transforms in a representation
of this symmetry which also determines the monodromy data of the theory
like braiding and fusing matrices. Hence we say that the CFT has the
monodromy data based on $\mathcal{D}$ and the bulk object $F$ describing
all bulk fields.

Since we are interested in correlators of bulk fields, we consider
conformal blocks that are based on the modular tensor category $\mathcal{D}=\mathcal{C}^{\rm rev}\boxtimes\mathcal{C}$
for a modular tensor category $\mathcal{C}$ (correlators of bulk
fields are obtained by combining conformal blocks for left movers
with those for right movers). Because of Proposition~\ref{prp:mtc},
we can replace $\mathcal{D}=\mathcal{C}^{\rm rev}\boxtimes\mathcal{C}$
with the Drinfeld center $Z(\mathcal{C})$ and therefore apply the
Turaev--Viro--Barrett--Westbury state-sum construction, or equivalently,
the string-net model described in Section~\ref{sec:String-net-model}.

We therefore define the \textit{pinned block functor}
$$
\mathrm{Bl}^{F}\colon\ \mathcal{S}\text{urf}\to\mathcal{V}\text{ect}_{\mathbb{C}}
$$
 by assigning to the extended surface $\Sigma_{p\mid q}^{g}$ the
finite dimensional vector space
$$
\mathrm{Bl}^{F}\big(\Sigma_{p\mid q}^{g}\big):= Z_{{\rm SN},\mathcal{C}}\big(\Sigma_{p+q}^{g},\underbrace{F^{\vee},\dots,F^{\vee}}_{p},
\underbrace{F,\dots,F}_{q}\big)
\cong\mathrm{Hom}_{Z(\mathcal{C})}\big(\mathbb{I}_{Z(\mathcal{C})},(F^{\vee})^{\otimes p}\otimes F^{\otimes q}\otimes L_{Z(\mathcal{C})}^{\otimes g}\big),
$$
 and to a morphism $(\varphi,s)$ between extended surfaces the natural
action of the mapping class $\varphi$ followed by the concatenation
of the string-net induced by the sewing $s$.

As an auxiliary datum, we also define the \textit{trivial block functor}
$\Delta_{\mathbb{C}}\colon S\text{urf}\to\mathcal{V}\text{ect}_{\mathbb{C}}$
by~assig\-ning to every extended surface the vector space $\mathbb{C}$
and to every morphism the identity $\mathrm{id}_{\mathbb{C}}$. \textit{A~consistent system of bulk field correlators} is then a monoidal natural
transformation
$$
v^{F}\colon\ \Delta_{\mathbb{C}}\to\mathrm{Bl}^{F}
$$
 such that $v^{F}\big(\Sigma_{1\mid1}^{0}\big)\coloneqq\big(v^{F}\big)_{\Sigma_{1\mid1}^{0}}(1)\in\mathrm{Bl}^{F}
 \big(\Sigma_{1\mid1}^{0}\big)\cong\mathrm{End}_{Z(\mathcal{C})}(F)$
is invertible.

Unpacking the rather compact definition above, we see that the so
defined consistent system of bulk field correlators amounts to a choice
of a vector
$$
v^{F}\big(\Sigma_{p\mid q}^{g}\big)\coloneqq\big(v^{F}\big)_{\Sigma_{p\mid q}^{g}}(1)\in\mathrm{Bl}^{F}\big(\Sigma_{p\mid q}^{g}\big)
$$
 for each extended surface $\Sigma_{p\mid q}^{g}$ that is invariant
under the action of the mapping class group $\mathrm{Map}\big(\Sigma_{p\mid q}^{g}\big)$,
such that the linear map induced by a sewing maps the chosen vector
to the chosen vector for the sewn surface.

It is shown~\cite[Theorem 4.10]{MR3590526} that for a (not necessarily
semisimple) modular finite category $\mathcal{D}$, the consistent
systems of bulk field correlators with monodromy data based on $\mathcal{D}$
and with bulk object $F\in\mathcal{D}$ are in bijection with structures
of a \textit{modular Frobenius algebra}~\cite[Definition~4.9]{MR3590526}
on $F$.

\subsection{Fundamental correlators via string-nets\label{subsec:Fund}}

Let $F\in Z(\mathcal{C})$ be a commutative symmetric Frobenius algebra. We~define the following string-nets on the set of surfaces that generates
all extended surfaces by sewing
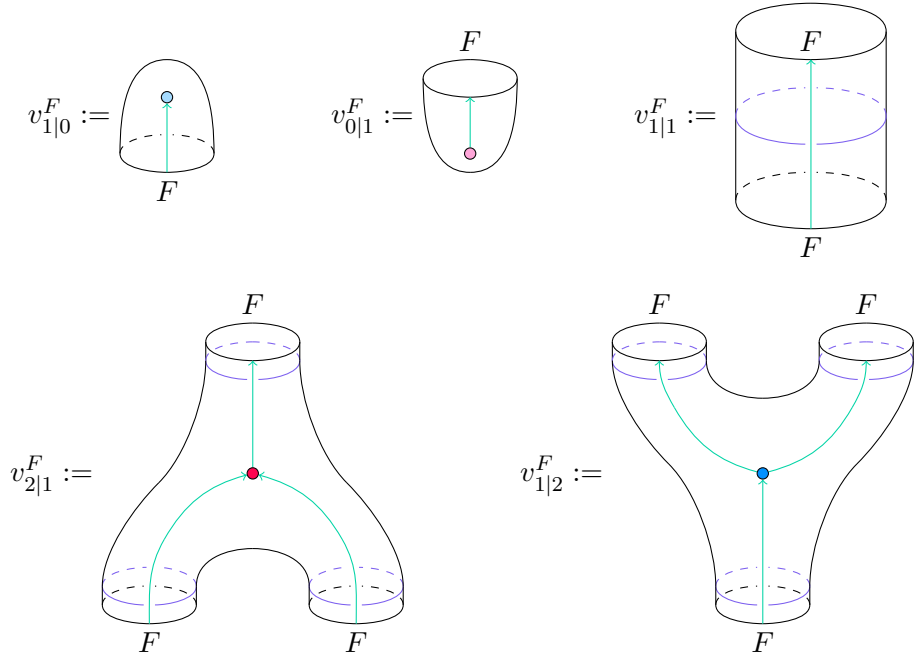
\begin{center}
$v_{1\mid0}^{F}\coloneqq\input{CS1.tikz}\qquad\qquad
v_{0\mid1}^{F}\coloneqq\input{CS2.tikz}\qquad\qquad
v_{1\mid1}^{F}\coloneqq\input{MC1.tikz}$
\end{center}

\begin{center}
$v_{2\mid1}^{F}\coloneqq\input{MP1.tikz}\qquad\qquad
v_{1\mid2}^{F}\coloneqq\input{MP2.tikz}$
\end{center}

\noindent Here the vertices are given by the coproduct of the Frobenius algebra~$F$. A priori, our prescription
depends on the isotopy classes of the embeddings of the string diagrams
into the surfaces. However, by using the following \textit{cloaking
relation}, we can show that the string-nets above are in fact well
defined:
\begin{Lemma}
\label{lem:cloaking}Let $\Sigma$ be a compact oriented surface,
$\mathbf{V}$ a boundary value, and $X,Y\in Z(\mathcal{C})$. We~have
the following equation in the string-net space $Z_{{\rm SN},\mathcal{C}}(\Sigma,\mathbf{V})$
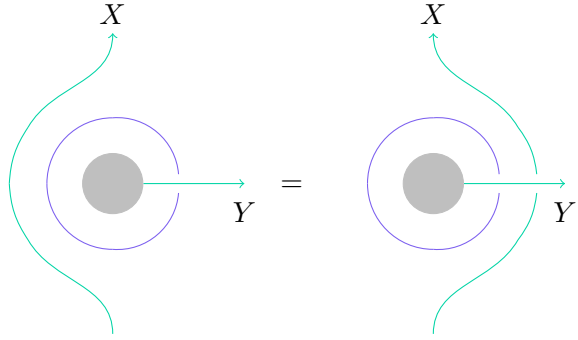
\begin{center}
$\input{CL1.tikz}\;=\;\input{CL2.tikz}$
\end{center}

\noindent here it is understood that the string-nets agree outside
the depicted region, and the shaded area may include boundary components.
\end{Lemma}

\begin{proof}
By using Proposition~\ref{prop:useful} and the naturality of the
half-braiding, we can see that both sides are equal to
\begin{center}
${\displaystyle \sum_{i,j\in\mathcal{I}(\mathcal{C})}}\dfrac{d_{i}d_{j}}{D^{2}}\input{CL3.tikz}$
\end{center}
\vspace{-10mm}
\end{proof}
\begin{Lemma}
\label{lem:F}
For a commutative symmetric Frobenius algebra $F\in Z(\mathcal{C})$,
$v_{1\mid0}^{F}$, $v_{0\mid1}^{F}$, $v_{1\mid1}^{F}$, $v_{2\mid1}^{F}$,
and $v_{1\mid2}^{F}$ are invariant under the mapping class groups.
\end{Lemma}

\begin{proof}
The cases of $v_{1\mid0}^{F}$, $v_{0\mid1}^{F}$ are trivial, since
the mapping class group of a disc is trivial.

The mapping class group of a cylinder is generated by a Dehn twist.
By doing a Dehn twist~$T$ along the projector on the cylinder, we
have:
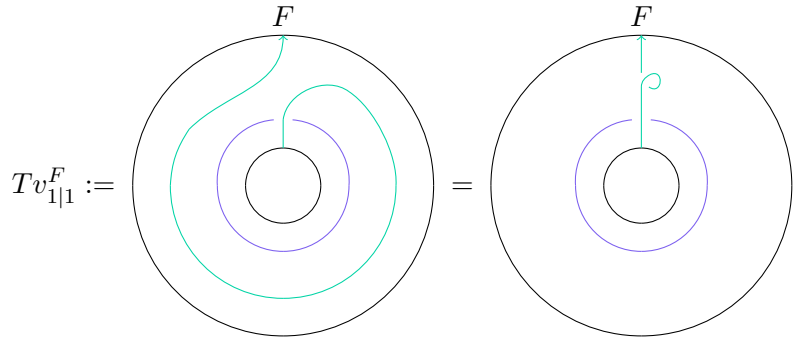
\begin{center}
$Tv_{1\mid1}^{F}\coloneqq\;\input{C1.tikz}\;=\;\input{C3.tikz}$
\end{center}

\noindent Here we have used the cloaking relation in Proposition~\ref{lem:cloaking}. Since being commutative and symmetric implies
that $F$ has trivial twist~\cite[Proposition 2.25]{MR2187404}, we
have $Tv_{1\mid1}^{F}=v_{1\mid1}^{F}$.

Now consider the so-called $B$ move on a pair of pants with two ingoing
boundary components, we have
\begin{center}
$Bv_{2\mid1}^{F}:=\input{MP3.tikz}=\input{M1.tikz}=\input{M2.tikz}=v_{2\mid1}^{F}$.
\end{center}

\noindent Here we have used the cloaking relation twice on the line
ending on the lower right circle and then the commutativity and symmetry
of $F$. Since the $B$ move and the Dehn twists along the boundary
circles generate the mapping class group, we have shown the invariance.
The~inva\-riance of $v_{1\mid2}^{F}$ can be shown in a similar way
by using cocommutativity of $F$.
\end{proof}
By setting $v^{F}\big(\Sigma_{1\mid0}^{0}\big)=v_{1\mid0}^{F}$, $v^{F}\big(\Sigma_{0\mid1}^{0}\big)=v_{0\mid1}^{F}$,
$v^{F}\big(\Sigma_{1\mid1}^{0}\big)=v_{1\mid1}^{F}$, $v^{F}\big(\Sigma_{2\mid1}^{0}\big)=v_{2\mid1}^{F}$,
$v^{F}\big(\Sigma_{1\mid2}^{0}\big)=v_{1\mid2}^{F}$, we can extend our prescription
to a consistent system of correlators via sewing, provided that the
string-net we get on the torus with one boundary component is invariant
under the action of the mapping class group. The argument is essentially
the same as the one used in~\cite{MR3590526}, i.e., via considering
the Lego--Teichm\"uller game~\cite{MR1780935}. Notice that the consistency
regarding surfaces of genus zero follows purely from the fact that~$F$ is a commutative symmetric Frobenius algebra in~$Z(\mathcal{C})$
and cloaking. For instance, since we can move the projectors around
by using the cloaking relation, the Frobenius property implies
\begin{center}
$v^{F}\big(\Sigma_{2\mid2}^{0}\big):=\scalebox{0.65}{\text{\LARGE{\input{CS7.tikz}}}}
=\scalebox{0.65}{\text{\LARGE{\input{CS8.tikz}}}}=\scalebox{0.65}{\text{\LARGE{\input{CS9.tikz}}}}$
\end{center}

It was shown in~\cite[Lemma~6.6]{traube2020cardy} that the condition
of $S$-invariance of the string-net on the torus with one boundary
circle corresponds to the $S$-invariance condition in~\cite[Lemma~3.2]{MR2551797},
which is equivalent to the modularity condition of the Frobenius algebra
given in~\cite[Definition~4.9]{MR3590526} in the semi-simple cases.
The surprising result~\cite[Theorem~3.4]{MR2551797} states that a
haploid commutative symmetric Frobenius algebra $F\in Z(\mathcal{C})$
is modular if and only if $\dim (F)=D^{2}$.

\section{The Cardy case}\label{sec:Cardy-case}

\subsection{The bulk algebra for the Cardy case}\label{subsec:Cardy}

So far, we have been working with a general commutative symmetric
Frobenius algebra. We~now consider a specific Frobenius algebra, which
is the algebra of the bulk fields in the Cardy case.

We equip the object
$$
L\coloneqq\bigoplus_{i\in\mathcal{I}(\mathcal{C})}X_{i}^{\vee}\otimes X_{i}
$$
with the under-over half-braiding $\gamma_{L}^{\rm uo}$ introduced at
the end of Section~\ref{subsec:Modular-tensor-categories} and
denote $F_{1}:=\big(L,\gamma_{L}^{\mathrm{uo}}\big)\in Z(\mathcal{C})$
in the following.

We next recall that it has a natural Frobenius algebra structure in
$Z(\mathcal{C})$, see a review~\cite[Section~2.2]{MR3797694}
and references therein.
\begin{Proposition}\label{Prop:frob}
$\big(F_{1},\mu_{F_{1}},\eta_{F_{1}},\Delta_{F_{1}},\varepsilon_{F_{1}}\big)$
is a commutative, symmetric Frobenius algebra in $Z(\mathcal{C})$,
where
\begin{center}
$\mu_{F_{1}}\coloneqq\bigoplus_{i,j,k\in\mathcal{I}(\mathcal{C})}d_{k}\scalebox{0.8}{\input{FS2.tikz}} \qquad \Delta_{F_{1}}\coloneqq\bigoplus_{i,j,k\in\mathcal{I}(\mathcal{C})}\dfrac{d_{j}d_{k}}{D^{2}}
\scalebox{0.8}{\input{FS6.tikz}}$
\end{center}
\begin{center}
$\eta_{F_{1}}\coloneqq\bigoplus_{i\in\mathcal{I}(\mathcal{C})}\delta_{0,i}\scalebox{0.8}{\input{HS7.tikz}}
\qquad \varepsilon_{F_{1}}\coloneqq\bigoplus_{i\in\mathcal{I}(\mathcal{C})}D^2\delta_{0,i}
\scalebox{0.8}{\input{HS11.tikz}}$
\end{center}
\end{Proposition}

In order to show that the prescription given in the Section~\ref{subsec:Fund}
for the Frobenius algebra $F_{1}$ extends to a consistent system
of correlators, we need to show that the string-net we get on the
torus with one boundary circle is invariant under the mapping class
group action. This is guaranteed by its dimension according to~\cite[Theorem 3.4]{MR2551797}.
However, the consistency for the Cardy case can be seen in a much
more straight forward and geometric manner, and a closed form of the
correlators can be derived: it turns out that the string-nets we get,
in their most simplified forms, are as empty as possible.

\subsection{Consistency made explicit}

The coend $L=\bigoplus_{i\in\mathcal{I}(\mathcal{C})}X_{i}^{\vee}\otimes X_{i}$
can be also equipped with a different half-braiding that can be understood
from the central monad
\begin{center}
$\gamma^{\mathrm{non}}_{L;X}\coloneqq{\displaystyle \bigoplus_{i,j\in\mathcal{I}(\mathcal{C})}}d_{j}\;\scalebox{0.8}{\input{CE1.tikz}}$
\end{center}

We call it the \textit{non-crossing half-braiding} for the
obvious reason.

We denote $\widetilde{F}\coloneqq\big(L,\gamma_{L}^{\mathrm{non}}\big)\in Z(\mathcal{C})$.
There is also a naturally defined Frobenius algebra structure on this
object, with the multiplication and co-multiplication given by
\begin{center}
$\scalebox{0.8}{\input{CE2.tikz}}\coloneqq\bigoplus_{i\in\mathcal{I}(\mathcal{C})}d_{i}^{-1}
\scalebox{0.8}{\input{CE3.tikz}} \qquad\qquad \scalebox{0.8}{\input{CE4.tikz}}\coloneqq\bigoplus_{i\in\mathcal{I}(\mathcal{C})}
\scalebox{0.8}{\input{CE5.tikz}}$
\end{center}

It is easy to show that this is a special symmetric Frobenius
algebra.
\begin{Proposition}
\label{thm:S}
For a modular tensor category $\mathcal{C}$, the morphism
$$
S_{L}\coloneqq{\displaystyle \bigoplus_{i,j\in\mathcal{I}(\mathcal{C})}d_{j}}\scalebox{0.8}{\input{CE6.tikz}}\in\mathrm{End}_{\mathcal{C}}(L)
$$
is an isomorphism of Frobenius algebras in $Z(\mathcal{C})$ from
the Cardy bulk algebra $\big(F_{1},\mu_{F_{1}},\eta_{F_{1}}$, $\Delta_{F_{1}},\varepsilon_{F_{1}}\big)$
in Proposition~$\ref{Prop:frob}$ to the Frobenius algebra $\big(\widetilde{F},\mu_{\widetilde{F}},\eta_{\widetilde{F}},
\Delta_{\widetilde{F}},\varepsilon_{\widetilde{F}}\big)$
defined above, with the inverse given by
$$
S_{L}^{-1}\coloneqq{\displaystyle \bigoplus_{i,j\in\mathcal{I}(\mathcal{C})}\frac{d_{j}}{D^2}}\scalebox{0.8}{\input{CE7.tikz}}
$$
\end{Proposition}

\begin{proof}
A more general form of the fact that the given morphisms are isomorphisms
of the Frobenius algebras (regarded as Frobenius algebras in $\mathcal{C}$)
was proven in~\cite[Proposition 4.3]{kong2008morita}. We~present
here a simple proof of the special case we need. Note that this can
be generalized to non-semisimple settings.

Using Proposition~\ref{prop:useful}, it is not hard to see that
$$
S_{L}\in\mathrm{Hom}_{Z(\mathcal{C})}\big(F_{1},\widetilde{F}\big) \qquad \text{and} \qquad S_{L}^{-1}\in\mathrm{Hom}_{Z(\mathcal{C})}\big(\widetilde{F},F_{1})\big.
$$
For instance
\begin{center}
$\scalebox{0.8}{\input{CE8.tikz}}=\bigoplus_{i,j\in\mathcal{I}(\mathcal{C})}d_{j}\scalebox{0.8}{\input{CE9.tikz}}
=\bigoplus_{i,j,k\in\mathcal{I}(\mathcal{C})}d_{j}d_{k}\scalebox{0.8}{\input{CE10.tikz}}=
\scalebox{0.8}{\input{CE11.tikz}}$
\end{center}

\noindent The fact that $S_{L}$ and $S_{L}^{-1}$ are inverse to
each other is equivalent to Lemma~\ref{lem:mtc3}.

To show that $S_{L}$ is an isomorphism of algebras, we notice
\begin{center}
$\scalebox{0.75}{\input{CE12.tikz}}={\displaystyle \bigoplus_{i,j,k,l,m,n\in\mathcal{I}(\mathcal{C})} \frac{d_{k}d_{l}d_{m}d_{n}}{D^4}}\scalebox{0.75}{\input{CE13.tikz}}
={\displaystyle \bigoplus_{i,j,k,l,n\in\mathcal{I}(\mathcal{C})} \frac{d_{k}d_{l}d_{n}}{D^4}}\scalebox{0.75}{\input{CE14.tikz}}$
\\[2ex]
$=\!\!{\displaystyle \bigoplus_{i,j,k,l,n\in\mathcal{I}(\mathcal{C})}\!\!\!\frac{d_{k}d_{l}d_{n}}{D^4}}
\scalebox{0.75}{\input{CE15.tikz}}
=\!\!{\displaystyle \bigoplus_{i,j,l\in\mathcal{I}(\mathcal{C})}\!\!\!\frac{d_{l}}{D^2}}\!\!
\scalebox{0.75}{\input{CE16.tikz}}
=\!\!{\displaystyle \bigoplus_{i\in\mathcal{I}(\mathcal{C})}d_{i}^{-1}}\!\!\scalebox{0.75} {\input{CE17.tikz}}=\!\!\scalebox{0.75}{\input{CE18.tikz}}$
\end{center}

\noindent Hence $S_{L}\circ\mu_{F}=\mu_{\widetilde{F}}\circ\big(S_{L}\otimes S_{L}\big)$.
Similarly, one shows that $S_{L}$ is also an isomorphism of coalgebras
\begin{center}
$\scalebox{0.75}{\input{CE19.tikz}}
\!\!=\!\!\!{\displaystyle \bigoplus_{i,j,k,l,m,n\in\mathcal{I}(\mathcal{C})}\!\!\!\!\frac{d_{j}d_{k}d_{l}d_{m}d_{n}}{D^4}}
\scalebox{0.75}{\input{CE20.tikz}}\!\!\!=\!\!\!{\displaystyle \bigoplus_{i,j,k,m,n\in\mathcal{I}(\mathcal{C})}\!\!\!\!\frac{d_{j}d_{k}d_{m}d_{n}}{D^4}}
\scalebox{0.75}{\input{CE21.tikz}}$
\\[1ex]
$={\displaystyle \bigoplus_{i,k,m\in\mathcal{I}(\mathcal{C})}\frac{d_{k}d_{m}}{D^2}}
\scalebox{0.75}{\input{CE22.tikz}}
={\displaystyle \bigoplus_{i,k,m\in\mathcal{I}(\mathcal{C})}\frac{d_{k}d_{m}}{D^2}}
\scalebox{0.75}{\input{CE23.tikz}}$
\\[1ex]
$={\displaystyle \bigoplus_{i,k,m\in\mathcal{I}(\mathcal{C})}\frac{d_{k}d_{m}}{D^2}}
\scalebox{0.75}{\input{CE24.tikz}}={\displaystyle \bigoplus_{i\in\mathcal{I}(\mathcal{C})}}\scalebox{0.75}{\input{CE25.tikz}}=\scalebox{0.75}{\input{CE26.tikz}}$
\end{center}
\vspace{-6mm}
\end{proof}

\begin{Corollary}
\label{cor:F tilda}
For a modular tensor category $\mathcal{C}$,
$\big(\widetilde{F},\mu_{\widetilde{F}},\eta_{\widetilde{F}}, \Delta_{\widetilde{F}},\varepsilon_{\widetilde{F}}\big)$
is a commutative, symmetric Frobenius algebra in $Z(\mathcal{C})$.
In particular, $\widetilde{F}$ has trivial twist.
\end{Corollary}

\begin{Remark}
In fact, Corollary~\ref{cor:F tilda} holds true for any spherical
fusion category $\mathcal{C}$.
\end{Remark}

It is a general fact that isomorphisms between bulk algebras induce
isomorphisms of the spaces of conformal blocks. In the case at hand,
this is implemented by composing the string-nets with the morphism
$S_{L}$ near the outgoing boundary and precomposing the string-net
with~$S_{L}^{-1}$ near the ingoing boundary. For instance, applying
to the invariants on pairs of pants, we~get
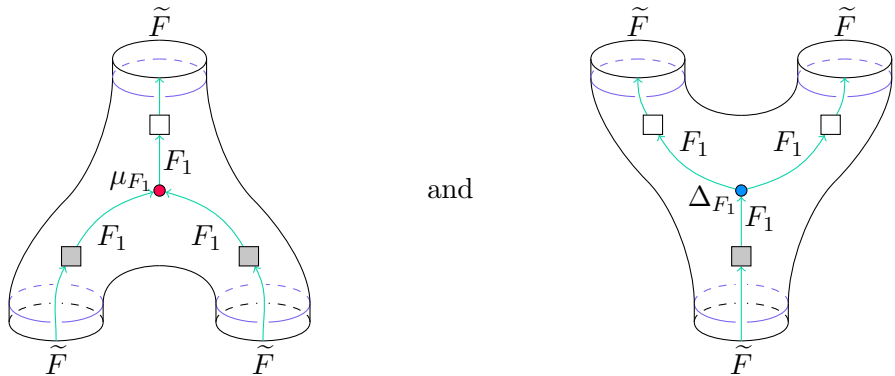
\begin{center}
$\input{CE27.tikz}\qquad\qquad\text{and}\qquad\qquad\input{CE29.tikz}$
\end{center}

\noindent Here the white boxes stand for $S_{L}$ and the gray ones
stand for $S_{L}^{-1}$. Since both are morphisms in $Z(\mathcal{C})$,
it makes no difference which side of the projectors we put the boxes
on, as long as we use the correct half-braidings.

If we take $\widetilde{F}$ as our bulk object, we get another set
of conformal blocks
$$
\mathrm{Bl}^{\widetilde{F}}\colon\ \mathcal{S}\text{urf}\to\mathcal{V}\text{ect}_{\mathbb{C}}
$$
 as well as a new set of correlators
$$
v^{\widetilde{F}}\colon\ \Delta_{\mathbb{C}}\to\mathrm{Bl}^{\widetilde{F}}.
$$

\noindent In fact, the induced isomorphisms of string-net spaces give
rise to a natural isomorphisms of~con\-for\-mal blocks
$$
\mathrm{Bl}^{S_{L}}\colon\ \mathrm{Bl}^{F_{1}}\to\mathrm{Bl}^{\widetilde{F}},
$$
 since the isomorphisms intertwine the action of mapping class groups
and sewing.

Moreover, due to the fact that
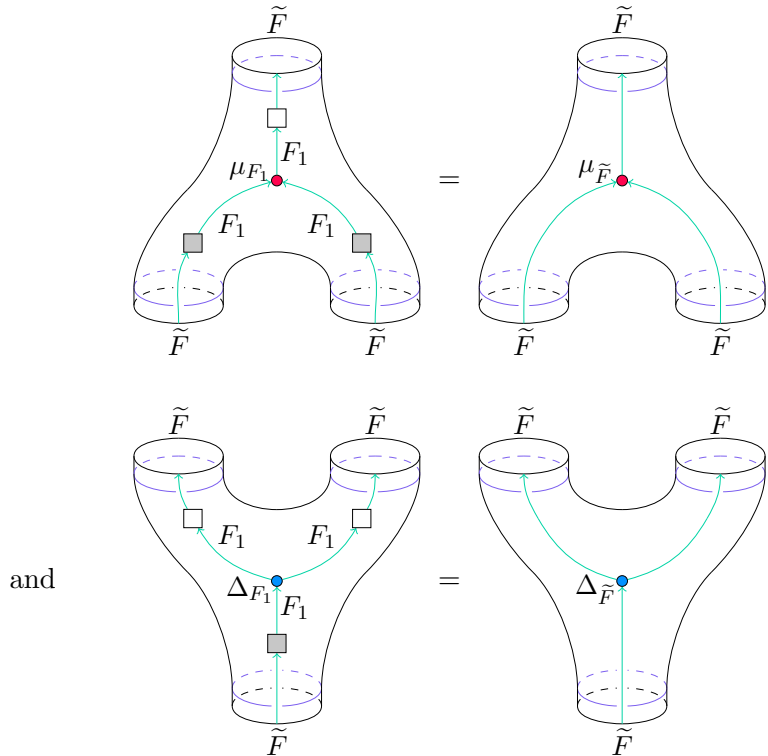
\begin{center}
{}\phantom{and \qquad}$\scalebox{0.95}{\input{CE27.tikz}}=\scalebox{0.95}{\input{CE28.tikz}}$
\end{center}

\noindent
\begin{center}
and \qquad $\scalebox{0.95}{\input{CE29.tikz}}=\scalebox{0.95}{\input{CE30.tikz}}$
\end{center}

\noindent we get a commutative diagram of natural transformations
\begin{center}
$\input{CE34.tikz}$
\end{center}

\noindent Intuitively, the two isomorphic Frobenius algebras produce
equivalent sets of correlators. The natural isomorphism $\mathrm{Bl}^{S_{L}}$
gives the precise way to relate them.

It turns out that the correlators given by the Frobenius algebra $\big(\widetilde{F},\mu_{\widetilde{F}},\eta_{\widetilde{F}},\Delta_{\widetilde{F}}, \varepsilon_{\widetilde{F}}\big)$
are particularly easy to compute:
\begin{Theorem}
\label{thm:main}Let $\Sigma_{p\mid q}^{g}$ be a surface of genus
$g$ with $p$ ingoing and $q$ outgoing boundaries, where $p,q,g\in\mathbb{Z}_{\geq0}$.

{\samepage
$1.$~The correlator for $\Sigma_{p\mid q}^{g}$ with bulk field $\widetilde{F}$
is given by the following string net:
\begin{center}
$v^{\widetilde{F}}\big(\Sigma_{p\mid q}^{g}\big)={\displaystyle \sum_{i_1,\dots,i_p,j_1,\dots,j_q\in\mathcal{I}(\mathcal{C})}}\dfrac{d_{j_1}\cdots d_{j_q}}{D^{2p}}\;\input{CE35.tikz}$
\end{center}

}

$2.$~Consequently, the correlator for $\Sigma_{p\mid q}^{g}$ with bulk
field $F_{1}$ is given by the following string net:
\begin{center}
$v^{F_{1}}\big(\Sigma_{p\mid q}^{g}\big)=\!\!\!{\displaystyle \sum_{\substack{i_1,\dots,i_p,j_1,\dots,j_q, \\ k_1,\dots,k_p,l_1,\dots,l_q\in\mathcal{I}(\mathcal{C})}}}\!\!\!\!\!\dfrac{d_{j_1}\cdots d_{j_q}d_{k_1}\cdots d_{k_p}d_{l_1}\cdots d_{l_q}}{D^{2(p+q)}}\;\input{CE40.tikz}$
\end{center}

\end{Theorem}

\begin{proof}
We only have to check the cases in which $g=0$ and $p+q\leq3$.
For $\Sigma_{2\mid1}^{0}$, we find
\begin{center}
$v^{\widetilde{F}}\big(\Sigma_{2\mid 1}^{0}\big)=\input{CE28.tikz}={\displaystyle \sum_{i,j,k,l,m,n,o\in\mathcal{I}(\mathcal{C})}\frac{d_{k}d_{l}d_{m}d_{n}d_{o}}{D^6}}\;\input{CE36.tikz}$\\$\phantom{A}$\\$={\displaystyle \sum_{i,j,k,l\in\mathcal{I}(\mathcal{C})}\frac{d_{k}d_{l}}{D^6}}\;\input{CE37.tikz}={\displaystyle \sum_{i,j,k\in\mathcal{I}(\mathcal{C})}\frac{d_{k}}{D^4}}\;\input{CE38.tikz}$
\end{center}

{\samepage\noindent Similarly, we have
\begin{center}
$v^{\widetilde{F}}\big(\Sigma_{1\mid 2}^{0}\big)=\input{CE30.tikz}={\displaystyle \sum_{i,j,k\in\mathcal{I}(\mathcal{C})}\frac{d_{j}d_{k}}{D^2}}\;\input{CE39.tikz}$
\end{center}

}

 The arguments concerning the unit and counit are even more straight
forward. Notice that, whenever we sew together a pair of boundaries,
we get a contractible circle that cancels out a~factor of $D^{2}$.

By applying the inverse of the natural isomorphism $\mathrm{Bl}^{S_{L}}$,
we get the second part of the statement.
\end{proof}
\begin{Remark}
It can be seen from the proof above, even though the definition of
the Frobenius algebra $F_{1}$ requires $\mathcal{C}$ to be braided
and the construction of the isomorphisms in Proposition~\ref{thm:S}
requires $\mathcal{C}$ to be modular, the Frobenius algebra $\widetilde{F}\in Z(\mathcal{C})$
gives rise to a consistent system of~bulk field correlators for any
spherical fusion category $\mathcal{C}$.
\end{Remark}

Theorem~\ref{thm:main} allows us to compute in particular the zero-point
correlator on a torus. Consider the following set of vectors $\{G_{i,j}\}_{i,j\in\mathcal{I}(\mathcal{C})}$
in the string-net space of a torus, where
\begin{center}
$G_{i,j}\coloneqq{\displaystyle \sum_{k\in\mathcal{I}(\mathcal{C})}\frac{d_{k}}{D^{2}}}\;\scalebox{0.65}{\text{\LARGE\input{MT3.tikz}}}$
\end{center}

\noindent here the opposite sides of the square are identified so
the resulting surface is a torus. For a~modular tensor category $\mathcal{C}$,
every simple object in $Z(\mathcal{C})\simeq\mathcal{C}^{\rm rev}\boxtimes\mathcal{C}$
is isomorphic to $Z_{(i,j)}\coloneqq\big(X_{i}\otimes X_{j},\gamma_{(i,j)}\big)$
for some $i,j\in\mathcal{I}(\mathcal{C})$ and the half braiding $\gamma_{(i,j)}$
given by the under-over half-braiding. It~can be seen from the following
representation of the string-net space associated to the torus
$$
{\displaystyle Z_{{\rm SN},\mathcal{C}}\big(\Sigma_{0}^{1}\big)\cong\bigoplus_{k\in\mathcal{I}(Z(\mathcal{C}))} Z_{{\rm SN},\mathcal{C}}\big(\Sigma_{2}^{0},Z_{k}^{\vee},Z_{k}\big) \cong \bigoplus_{i,j\in\mathcal{I}(\mathcal{C})}Z_{{\rm SN},\mathcal{C}}\big(\Sigma_{2}^{0},Z_{(i,j)}^{\vee},Z_{(i,j)}\big)}
$$
that follows from factorization that $\{G_{i,j}\}_{i,j\in\mathcal{I}(\mathcal{C})}$, up to the action of the mapping class group $\mathrm{Map}\big(\Sigma_{0\mid0}^{1}\big)\cong\mathrm{SL}(2,\mathbb{Z})$,
is a basis for the vector space $Z_{{\rm SN},\mathcal{C}}\big(\Sigma_{0}^{1}\big)$.
As a result of Theorem~\ref{thm:main}, the torus partition function
$v^{F_{1}}\big(\Sigma_{0\mid0}^{1}\big)$ is the \textit{empty string-net},
which is obviously invariant under the mapping class group action.
When written in the following form
\begin{center}
$v^{F_{1}}\big(\Sigma_{0\mid0}^{1}\big)\coloneqq\input{MT1.tikz}={\displaystyle \sum_{i,j\in\mathcal{I}(\mathcal{C})}\frac{d_{j}}{D^{2}}}\scalebox{0.65}{\text{\LARGE\input{MT2.tikz}}}
={\displaystyle \sum_{i\in\mathcal{I}(\mathcal{C})}G_{\bar{i},i}}$,
\end{center}

\noindent the correlator is expressed as a linear combination of
the basis vectors $\{G_{i,j}\}_{i,j\in\mathcal{I}(\mathcal{C})}$
with the coefficients $(\delta_{i,\bar{j}})_{i,j\in\mathcal{I}(\mathcal{C})}$,
which are the entries of the charge conjugation matrix.

\subsection*{Acknowledgements}
 We thank Alain Brugui\`eres, J\"urgen Fuchs, Eilind Karlsson and Vincentas Mulevi\v{c}ius for helpful discussions. The authors are partially supported by the RTG 1670 ``Mathematics inspired by String theory and Quantum Field Theory'' and by the Deutsche Forschungsgemeinschaft (DFG, German Research Foundation) under Germany's Excellence Strategy -- EXC 2121 ``Quantum Universe''-QT.2.

\pdfbookmark[1]{References}{ref}
\LastPageEnding

\end{document}